\documentclass[12pt]{article}
\usepackage[colorlinks, citecolor={blue}]{hyperref}
\usepackage{amsmath,latexsym,amsbsy,amssymb,amsthm}
\usepackage{url}
\usepackage{color}
\usepackage[textsize=small]{todonotes}

\usepackage[margin=1in]{geometry}
\usepackage{setspace}
\usepackage{subcaption}
\usepackage{hyphenat}

\numberwithin{equation}{section}
\theoremstyle{plain}

\newcommand{\h}{\hspace*{.24in}}
\newtheorem{Assumption}{Assumption}[section]
\newtheorem{Remark}{Remark}[section]
\newtheorem{Definition}{Definition}[section]
\newtheorem{Lemma}{Lemma}[section]
\newtheorem{Theorem}{Theorem}[section]
\newtheorem{Corollary}{Corollary}[section]

\title{The shape of the one-dimensional phylogenetic likelihood function}
\date{}
\author{
Vu Dinh and Frederick A.~Matsen IV \\
Program in Computational Biology\\
Fred Hutchinson Cancer Research Center
}

\begin{document}
\maketitle

\begin{abstract}
By fixing all parameters in  a phylogenetic likelihood model except for one branch length, one obtains a one\hyp dimensional likelihood function.
In this work, we introduce a mathematical framework to characterize the shapes of such one\hyp dimensional phylogenetic likelihood functions.
This framework is based on analyses of algebraic structures on the space of all frequency patterns with respect to a polynomial representation of the likelihood functions.
Using this framework, we provide conditions under which the one\hyp dimensional phylogenetic likelihood functions are guaranteed to have at most one stationary point, and this point is the maximum likelihood branch length.
These conditions are satisfied by common simple models including all binary models, the Jukes-Cantor model and the Felsenstein 1981 model.

We then prove that for the simplest model that does not satisfy our conditions, namely, the Kimura 2-parameter model, the one\hyp dimensional likelihood functions may have multiple stationary points.
As a proof of concept, we construct a non-degenerate example in which the phylogenetic likelihood function has two local maxima and a local minimum.
To construct such examples, we derive a general method of constructing a tree and sequence data with a specified frequency pattern at the root.
We then extend the result to prove that the space of all rescaled and translated one\hyp dimensional phylogenetic likelihood functions under the Kimura 2-parameter model is dense in the space of all non-negative continuous functions on $[0, \infty)$ with finite limits.
These results indicate that one\hyp dimensional likelihood functions under advanced evolutionary models can be more complex than it is typically assumed by phylogenetic inference algorithms; however, these complexities can be effectively captured by the Kimura 2-parameter model.


\end{abstract}

\section{Introduction}

The likelihood of a phylogenetic model is a function of the parameters of continuous time Markov chains (CTMCs) used to model sequence evolution along each branch.
It is common to assume a single rate matrix and stationary frequency for the CTMCs but allow the branch lengths to vary, representing a single evolutionary process but differing amounts of evolution along each branch.
Commonly used maximum-likelihood phylogeny programs improve likelihood by modifying branch lengths iteratively and one at a time \cite{bryant2005likelihood}.
The general approach for numerical maximization of the one\hyp dimensional likelihood function given by fixing every parameter except for one branch length is to iteratively sample the function at a number of points, use surrogate functions to fit simple curves to those points, and use those fits as approximations to locate the maximum branch length.
For example, programs often employ Newton's method, in which the intuitive idea is to use first and second derivatives to approximate the likelihood function (varying along that branch) by a surrogate quadratic function.
Since evaluations of the likelihoods (and their derivatives) are computationally expensive, many approaches have been tried to improve the efficiency of this optimization procedure \cite{bryant2005likelihood}.

Such approaches, however, rely on the assumptions that one\hyp dimensional phylogenetic likelihood functions belong to some class of simple functions, and that the surrogate model can, at least, capture the shape of the functions.
While there has been a considerable amount of work on finding multiple maxima of the multi-dimensional likelihood surfaces parameterized by all branch lengths for a tree \cite{steel1994maximum,chor2000multiple, rogers1999multiple}, little has been done about the shapes of one\hyp dimensional phylogenetic likelihood functions.
The only attempt to investigate the shape of the one\hyp dimensional phylogenetic likelihood functions has been \cite{fukami1989maximum}, which provided a proof of uniqueness of the stationary points for one\hyp dimensional phylogenetic likelihood functions in the case of the one parameter model of nucleotide substitution.
Based on this proof, the authors of \cite{fukami1989maximum} asserted that there is at most one stationary point of the full likelihood surface.
This claim was later disproved by \cite{steel1994maximum}, although the proof for the one\hyp dimensional case still holds.
However, the result has not been examined for the more complex models used in practice.

In this work, we introduce a mathematical framework to characterize the shapes of such one\hyp dimensional phylogenetic likelihood functions.
This framework is based on analyses of algebraic structures on the space of all frequency patterns with respect to a polynomial representation of the likelihood functions.
Specifically, we introduce the new concept of \emph{logarithmic relative frequency patterns} and analyze algebraic structures on the space of such patterns.
These structures, along with the \emph{characteristic polynomial representations} of one\hyp dimensional phylogenetic likelihood functions, open a new way to explore the space of all possible likelihood functions.
Moreover, by composing these structures, we are able to tackle the inverse problem of constructing a phylogenetic tree that has a given frequency pattern at the root.
This enables us to construct phylogenetic trees that approximate any given likelihood function with arbitrary precision.

Using this framework, we provide conditions under which the one\hyp dimensional phylogenetic likelihood functions are guaranteed to have at most one stationary point, and this point is the maximum of the one\hyp dimensional function.
These conditions are satisfied by common simple models including all binary models, the Jukes-Cantor model \cite{jukes1969evolution} and the Felsenstein 1981 model \cite{felsenstein1981evolutionary}.
We then prove that for the simplest model that does not satisfy our conditions, namely, the Kimura 2-parameter model \cite{kimura1980simple}, the one\hyp dimensional likelihood functions may have multiple stationary points.
As a proof of concept, we construct a non-degenerate example in which the phylogenetic likelihood function has two local maxima and a local minimum.

We then extend the result to prove that the space of all rescaled and translated one\hyp dimensional phylogenetic likelihood functions under the Kimura 2-parameter model is dense in the set of all non-negative continuous functions on $[0, \infty)$ with finite limits.
These results indicates that one\hyp dimensional likelihood functions under advanced evolutionary models can be more complex than it is typically assumed by phylogenetic inference algorithms; however, these complexities can be effectively captured by the Kimura 2-parameter model.

\section{Background and Definitions}
\label{background}
\subsection{Markov models of sequence evolution}
Our setting is the standard IID setting for likelihood-based phylogenetics with a finite number of sites; we review the basics here but refer the reader to \cite{felsenstein2004inferring} for more details.
Let $\Omega$ denote the set of states and let $r= |\Omega|$.
For convenience, we assume that the states have indices $1$ to $r$.

For an unrooted tree $T$ with $N$ taxa, we use $E(T)$ and $V(T)$ to denote the set of edges and vertices of $T$, respectively.
On each edge $e \in E(T)$, we assume that the mutation events occur according to a continuous time Markov chain on states $\Omega$ with instantaneous rate matrix $Q_e$.
This rate matrix $Q_e$ and the branch length $t_e$ on the edge $e$ define the transition matrix $P^e=e^{Q_et_e}$ on edge $e$, where $P^e_{ij}(t_e)$ denotes the probability of mutating from state $i$ to state $j$ across the edge $e$ (with length $t_e$).

\begin{figure}
\centering
\begin{subfigure}{.5\textwidth}
  \centering
  \includegraphics[width=0.6\linewidth]{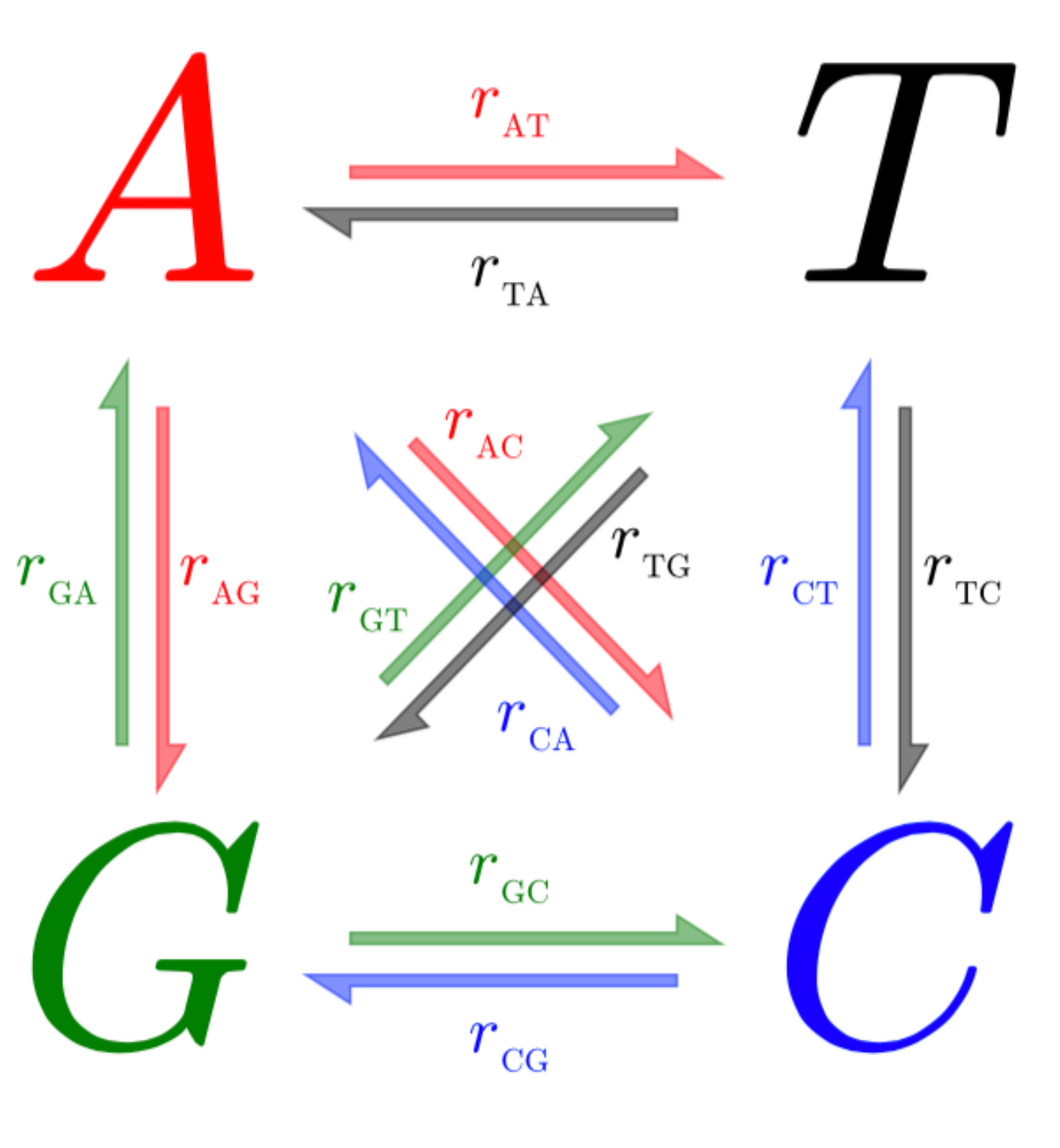}
\end{subfigure}%
\begin{subfigure}{.6\textwidth}
  \centering
  \includegraphics[width=0.8\linewidth]{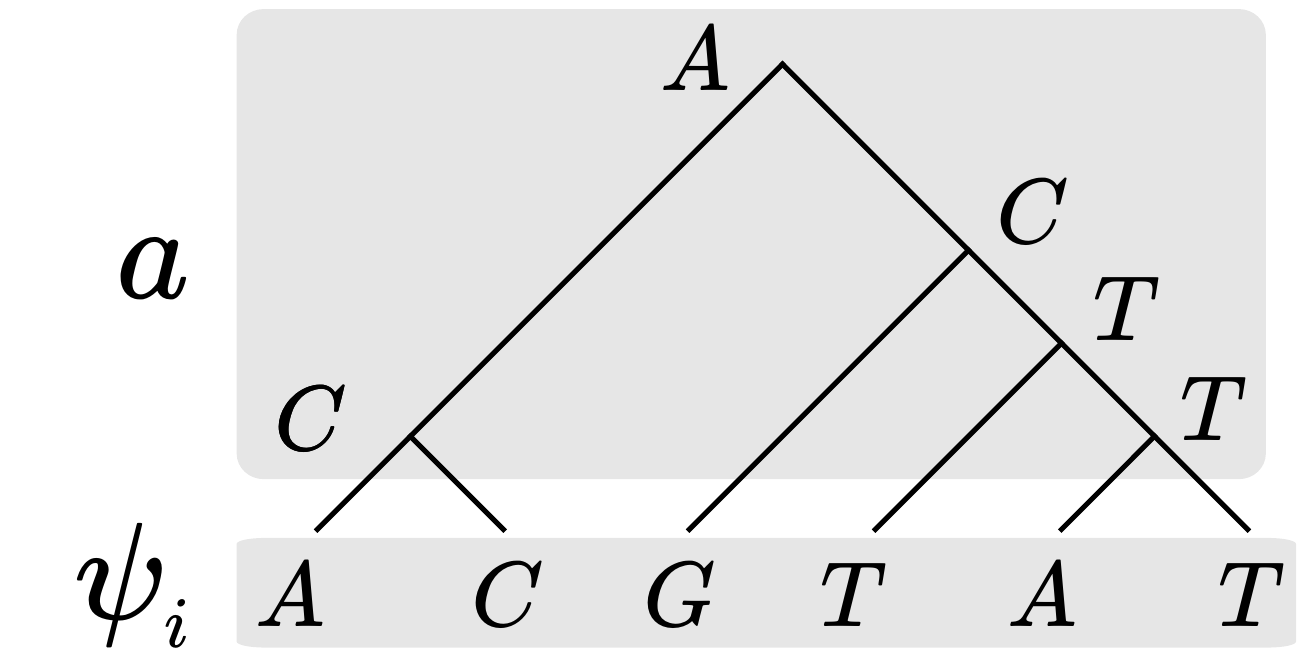}
\end{subfigure}
\caption{[Left:] A general Markov model of DNA evolution along a tree edge. [Right:] An extension $a$ of the labelling $\psi_i$ (corresponds to site $i$ in the sequences) of the leaves of a simple tree $\tau$ to its inner nodes. }
\label{Fig1}
\end{figure}

We further assume that for all edges $e \in E(T)$, the Markov chains that describe the mutation events are ergodic and time-reversible with respect to a fixed stationary distribution $\pi$, that is
\[
\lim_{t \to \infty} {P^e_{ij}(t)} = \pi_j,
\]
and
\[
\pi_i P^e_{ij} (t) = \pi_j P^e_{ji} (t) \h \forall t,
\]
for all $i, j \in \Omega$ and $e \in E(T)$.

The phylogenetic likelihood is computed as follows given a set of (aligned) observed sequences $\psi = (\psi_1, \psi_2,...,\psi_S) \in \Omega^{N \times S}$ of length $S$ over $N$ taxa of a tree $\tau$.
First orient the edges of $\tau$ away from an arbitrarily chosen root $\rho$ of the tree.
(We can choose the root arbitrarily since each $P_e$ is reversible with respect to $\pi$.)
Each site $i$ in the sequences determines a labeling $\psi_i$ of each leaf by a state in $\Omega$.
An extension $a$ of a labeling $\psi_i$ is an assignment of states to all of the nodes in the tree that agrees with $\psi$ on the leaves.

The probability of an extension $a$ given the vector of branch lengths $\bold{t} =(t_e)_{e \in E(T)}$ is defined to be the probability of the state at the root (given by the stationary distribution) multiplied by the probabilities of all the state transitions (including self-transitions) across each branch in the tree
\[
P(a | \mathbf{t}) =  \pi(a_{\rho}) \prod_{(u,v)\in E(T)}{P^{uv}_{a_ua_v}( t_{uv})},
\]
where $a_u$ denotes the assigned state of node $u$ by $a$.

The likelihood of the data at site $i$ is then the marginal probability over all the extensions
\[
P(\psi_i | \mathbf{t}) =  \sum_{a ~ \text{extends} ~\psi}{P(a | \mathbf{t})}.
\]
We further assume, as is standard, that evolution is independent between sites.
This implies that the likelihood of a set of sequences evolving is just the product of the probabilities for the individual sites
\[
L(\psi|\mathbf{t}) =  \prod_{s=1}^S{P(\psi_i | \mathbf{t})}.
\]
In summary, the likelihood of observing $\psi$ given the tree topology $\tau$ and the vector of branch lengths $\bold{t} =(t_e)_{e \in E(T)}$ has the form
\[
L(\psi|\mathbf{t}) = \prod_{s=1}^S{\sum_{a }{~ \pi(a_{\rho}) \prod_{(u,v)\in E(T)}{P^{uv}_{a_ua_v}( t_{uv})}}}
\]
where $a$ ranges over all extensions of $\psi$ to the internal nodes of $T$ and $a_u$ denotes the assigned state of node $u$ by $a$.

For readers familiar with the theory of probabilistic inference on graphical models, the likelihood functions studied in this paper can be alternatively described as follows.
Consider a tree $T$ and let $\{X_v: v \in V(T)\}$ be a collection of random variables indexed by the nodes of the tree.
For each edge $(u,v) \in E(T)$, we define the nonnegative potential function
\[
k_{(u,v)}(i, j, t) := P^{uv}_{ij}(t).
\]
We assume that the joint probability distribution $p(x_{V(T)})$ factorizes over the tree edges:
\[
p\left(x_{V(T)}\right) \sim \prod_{(u,v) \in E(T)}{k_{(u,v)}(x_u, x_v, t_{uv})}.
\]
The likelihood functions of interest may then be represented as the marginal probability of the observation $\psi$ on the leaves of the tree $T$.
This formulation allows us to study the phylogenetic likelihood functions beyond the reversible Markov framework.
We will investigate partial extensions to this more general case in Section $\ref{graphical}$, but for the next several sections we will focus on the standard phylogenetic setting (in which we can prove the strongest results).

\subsection{One-dimensional phylogenetic likelihood functions}

To investigate the one\hyp dimensional likelihood function on one branch $e_0$, we fix all other branches, partition the set of all extensions of $\psi$ according to their labels at the end points of $e_0$, and split $E(T)$ into two sets of edges $E_{\operatorname{left}}$ and $E_{\operatorname{right}}$ corresponding to the location of the edges with respect to $e_0$.
The likelihood function can be rewritten as a univariate function of $t$, the branch length of $e_0$:
\begin{eqnarray*}
L(\psi|t) &=& \prod_{s=1}^S{\sum_{ij}{\sum_{a \in \mathcal{A}_{ij}}{ ~ \pi(a_{\rho}) \left( \prod_{e \in E_{\operatorname{left}} }{P^e_{a_ua_v}( t_{uv})}\right) }}}\\
&& \h \h   \times ~ P^{e_0}_{ij}(t) \times \left( \prod_{e \in E_{\operatorname{right}} }{P^e_{a_ua_v}( t_{uv})}\right)
\end{eqnarray*}
where $\mathcal{A}_{ij}$ denotes the set of all extensions of $\psi$ for which the labels at the left end point and the right end point of $e_0$ are $i$ and $j$, respectively.
We note that some $\mathcal{A}_{ij}$ may be empty if $e_0$ is a pendant edge and the observed value on the corresponding leaf is not $i$.

By grouping the products over $E_{\operatorname{left}}$ and $E_{\operatorname{right}}$ as well as the sum over $a$ in a single term $b_{ij}^s$, we can define the one\hyp dimensional log-likelihood function as
\[
\ell_{e_0}(t)=\log{L(\psi|t)}=\sum_{s=1}^S{\log{\left(\sum_{ij}{b^s_{ij}P^{e_0}_{ij}(t)}\right)}}.
\]
Such $\ell_{e_0}(t)$ are the object of study of this paper.

For convenience, we will assume that $e_0$ has been chosen and will drop the index $e_0$ hereafter.

\subsection{Evolutionary models}

Throughout the paper, we use the term \emph{evolutionary model} on state set $\Omega$ to refer to a collection $\mathcal{H}$ of $(Q, \pi)$ pairs, where $\pi$ is a vector of stationary frequencies  and $Q$ is a rate matrix on $\Omega$ that is reversible with respect to $\pi$.
If at every edge of the tree $\tau$, the matrix-frequency pair $(Q_e, \pi)$ belongs to $\mathcal{H}$, we say that $\tau$ is a tree under evolutionary model $\mathcal{H}$.

We will consider a number of different evolutionary models of DNA sequences.
These DNA substitution models differ in terms of the parameters used to describe the rates at which one state replaces another during evolution and the stationary frequencies:
\begin{itemize}
\item Jukes-Cantor model \cite{jukes1969evolution}: this model assumes equal stationary frequencies ($\pi_A = \pi_G = \pi_T=\pi_C = 1/4$) and equal mutation rates.
\item Felsenstein 1981 model \cite{felsenstein1981evolutionary}: this is an extension of the Jukes-Cantor model in which  stationary frequencies are allowed to vary.
\item Kimura 2-parameter model \cite{kimura1980simple}: this model assumes equal stationary frequencies, but  distinguishes between the rates of transitions ($A \leftrightarrow G$, i.e. from purine to purine, or $C \leftrightarrow T$, i.e. from pyrimidine to pyrimidine) and transversions (from purine to pyrimidine or vice versa).

Following common usage,  we use $\kappa$ to denote the transition/transversion rate ratio and write the rate matrix for this model as
\begin{equation*}
Q_{\kappa}=\frac{1}{2(\kappa +1)}\left( \begin{array}{cccc}
-(\kappa+2) & \kappa & 1 & 1 \\
\kappa & -(\kappa+2) & 1 & 1 \\ \label{q0}
1 & 1 & -(\kappa+2) & \kappa \\
1 & 1 & \kappa & -(\kappa+2) \end{array} \right).
\end{equation*}
The special case $\kappa=3$ will play a central role in the analysis of this paper.
Note that the single $\kappa$ parameter in the Kimura 2-parameter model determines a rate matrix that is shared across the tree, while this paper primarily concerns the effect of changing a single branch length parameter.
\end{itemize}

While the focus here is on DNA models, we emphasize that our theoretical framework is capable of analyzing any time-reversible evolutionary model on any state space.
In fact, we do not assume a uniform molecular clock, or even a single evolutionary model along the edges of the tree.

\subsection{Characteristic polynomials of one\hyp dimensional phylogenetic likelihood functions}

We will frequently use the following assumption:
\begin{Assumption}
The eigenvalues of the rate matrix $Q$ are equal to
\begin{equation*}
0=d_0 \gamma \ge -d_1\gamma \ge -d_2\gamma \ge \ldots \ge -d_{r-1}\gamma
\end{equation*}
for some positive number $\gamma$ and non-negative integers $d_1, \ldots, d_{r-1}$.
\label{as1}
\end{Assumption}

The following remark, whose proof is provided in the Appendix, guarantees that Assumption~$\ref{as1}$ does not affect the generality of our analyses up to an arbitrarily small approximation error:

\begin{Remark}
The set of rate matrices $Q$ for a given evolutionary model that satisfy Assumption~$\ref{as1}$ is dense in the set of rate matrices under the same evolutionary model.
\label{rem1}
\end{Remark}

Under Assumption~$\ref{as1}$, if we denote the entries of the diagonalizing matrix $M$ and $N$ of $Q$ by $m_{ij}$ and $n_{ij}$, respectively, then the transition probabilities can computed as
\[
P_{ij}(t) = \sum_{k}{m_{ik} e^{-d_k \gamma t} n_{kj}}.
\]
By reparametrizing with $x:=e^{-\gamma t}$, we can represent these transition probabilities as polynomial functions
\[
P_{ij}(x) = \sum_{k}{m_{ik} x^{d_k} n_{kj}}.
\]
Similarly, the log-likelihood function can be rewritten as
\[
\ell(x)= \sum_{s=1}^S{\log{\left( \lambda_s(x)
\right)}} \h \text{where} \h
\lambda_s(x)= \sum_{ij}{b_{ij}^s P_{ij}(x)} .
\]
Hereafter, we will refer to $P_{ij}(x)$ and $\lambda_s(x)$ as the \textit{transition polynomials} of the evolutionary model and the \textit{characteristic polynomials} of the one\hyp dimensional phylogenetic likelihood function, respectively.

As we will see in later sections, this polynomial representation will enable us to exploit many algebraic and analytic properties of the likelihood functions.
The most noticeable feature is that one can use the Fundamental Theorem of Algebra to factorize $\lambda_s(x)$ as products of linear and quadratic polynomials.
As a result, the log-likelihood function can be written in the form
\begin{eqnarray*}
\ell(x) &=&  \sum_{s=1}^S \sum_{i=1}^{i_{s,1}}{\log(\alpha_{s,i} + \beta_{s,i} x)} \\
&& \h + \h  \sum_{s=1}^S \sum_{i=1}^{i_{s,2}}{\log(\mu_{s,i} + \nu_{s,i} x+ \omega_{s,i}x^2)}
\end{eqnarray*}
where $\mu_{s,i},  \nu_{s,i}, \omega_{s,i}$ are the (real) coefficients of the quadratic polynomials in the decomposition of $\lambda_s$, while $\alpha_{s,i}, \beta_{s,i}$ are coefficients of the linear terms in the decomposition.

This enables us to decompose a complicated evolutionary model into smaller modules, each of which can be approximated either by a ``linear'' model (like the binary symmetric model) or by a ``quadratic'' model (like the Kimura 2-parameter model).
In Section 3, we use this formulation to prove that if the phylogenetic log-likelihood function is essentially linear (that is, there are no quadratic terms in the expression), its shape resembles those generated by binary models, with a unique stationary point that is also the maximum point.
In Section 5, we illustrate that this property does not hold for quadratic models by constructing a counter-example with the Kimura 2-parameter model.
Finally, in Section 6, we use this formulation once again to prove that  the space of all rescaled and translated one\hyp dimensional phylogenetic likelihood functions under the Kimura 2-parameter model is dense in the space of all continuous functions on $[0, \infty)$ with finite limits.

\section{Uniqueness of the stationary point}
\label{sec:unique}

In this section, we discuss a condition under which the uniqueness of the stationary branch length is guaranteed.

The analyses in this section stem from two observations:
\begin{itemize}
\item[1. ] If for every site index$~s$, the characteristic polynomial $\lambda_s$ has no non-real root, then the likelihood function can be decomposed into smaller modules, each of which resemble a binary model.
\item[2. ] The likelihood functions of binary models and summations of such models are \emph{incave}.
\end{itemize}

\begin{Definition}[Hanson \cite{hanson1981sufficiency}]
A vector-valued function $f$ is said to be incave in $\mathbb{R}^n$ if there exists a vector-valued function $\phi(t,u)$ such that
\[
f(t)-f(u) \le \phi(t,u) \cdot \nabla f(u), \h \forall t, u \in \mathbb{R}^n
\]
where $\nabla f$ denotes the gradient of $f$.
\end{Definition}

Incave functions were introduced in the optimization literature as a generalization of concave functions\cite{hanson1981sufficiency}.
It can be proven that a function is incave if and only if every stationary point is a global maximum \cite{ben1986invexity}.
We are interested in the case of functions of a single real argument, for which the following result also holds:
\begin{Lemma}
If $f$ is  a real-valued incave function with a finite number of stationary points, then $f$ has at most one stationary point. Moreover, if such a point exists, it is also a global maximum.
\label{finite}
\end{Lemma}
\begin{proof} Denote $A=\{t \in [0, \infty): f'(t)=0\}$ and assume that $A$ has more than one element.
Since $A$ is finite, we can choose two elements $t_1$ and $t_2$ in $A$ such that the interval $(t_1,t_2) \subset \mathbb{R} - A$.
Since $f$ is incave, every stationary point of $f$ is a global maximum.
We deduce that $t_1$ and $t_2$ are both global maxima of $f$ and $f(t_1)=f(t_2)$.
Using the mean value theorem, there exists $t \in (t_1,t_2)$ such that $f'(t)=0$. This is a contradiction.
\end{proof}

This enables us to prove the following theorem.

\begin{Theorem}
If for every site index $s$, the polynomial $\lambda_s$ has only real roots, then $\ell$ has at most one stationary point. Moreover, if such a point exists, it is also a global maximum.
\label{Theo1}
\end{Theorem}

\begin{proof}
Since $\lambda_s$ has only real roots, it can be written as product of linear functions
\[
\lambda_s(x) = \prod_{i=1}^{d_p}{(\alpha_{s,i} + \beta_{s,i} x)}
\]
where $d_p$, defined in Assumption~$\ref{as1}$, is the degree of the polynomial $\lambda_s$.

The log-likelihood function $\ell$ can be computed as
\begin{eqnarray*}
\ell(t)&=& \sum_{s=1}^S{\log{\left( \lambda_s(e^{-\gamma t})
\right)}}\\
&=& \sum_{s=1}^S{\log{\left(\prod_{i=1}^{d_p}{(\alpha_{s,i} + \beta_{s,i} e^{-\gamma t})}\right)}}\\
&=& \sum_{s=1}^S{{\sum_{i=1}^{d_p}{\log(\alpha_{s,i} + \beta_{s,i} e^{-\gamma t})}}}.
\end{eqnarray*}

For any $t,u>0$, we have
\begin{eqnarray*}
\ell(t)-\ell(u) &=& \sum_{s=1}^S{{\sum_{i=1}^{d_p}{\log\left(\frac{\alpha_{s,i} + \beta_{s,i} e^{-\gamma t}}{\alpha_{s,i} + \beta_{s,i} e^{-\gamma u}} \right)  }}} \\
&\le& \sum_{s=1}^S{{\sum_{i=1}^{d_p}{\left(\frac{\alpha_{s,i} + \beta_{s,i} e^{-\gamma t}}{\alpha_{s,i} + \beta_{s,i} e^{-\gamma u}} -1 \right)  }}}  \\
&=& \sum_{s=1}^m{{\sum_{i=1}^{d_p}{\left(\frac{\beta_{s,i} (e^{-\gamma t} - e^{-\gamma u} )  }{\alpha_{s,i} + \beta_{s,i} e^{-\gamma u}} \right)  }}}  \\
&=& \frac{1}{\gamma}\left(1-e^{-\gamma (t-u)}\right) \sum_{s=1}^S{{\sum_{i=1}^{d_p}{\left(\frac{-\beta_{s,i} ~\gamma e^{-\gamma u}  }{\alpha_{s,i} + \beta_{s,i} e^{-\gamma u}} \right)  }}}  \\
&=& \frac{1}{\gamma}\left(1-e^{-\gamma (t-u)}\right)  \ell'(u).
\end{eqnarray*}
Hence, $\ell$ is an incave function.

Furthermore, since $\lambda_s$ are polynomial and $e^{-\gamma t}$ is a bijective map from $[0, \infty)$ to $(0,1]$, we deduce that $\ell(t)$ only has a finite number of stationary points.
Using Lemma~\ref{finite}, we conclude that $\ell$ has at most one stationary point; moreover, if such a point exists, it is also a global maximum.

\end{proof}

We note that Theorem~$\ref{Theo1}$ imposes a condition on the characteristic polynomials rather than the evolutionary model, and can be applied to assess the uniqueness of the stationary point of any time-reversible evolutionary model satisfying Assumption~$\ref{as1}$.
In fact, Theorem~$\ref{Theo1}$ does not assume a uniform molecular clock, or even a single evolutionary model along the edges of the tree. However, it is worth noting that for the class of models on which the rate matrices have only one non-zero eigenvalue, the result automatically holds:

\begin{Corollary}
For binary, Jukes-Cantor and Felsenstein 1981 models, the one\hyp dimensional likelihood function has at most one stationary point; if such point exists, it is the global maximum.
\label{model}
\end{Corollary}

We also note that the results in previous studies about the number of maxima of likelihood surfaces  \cite{steel1994maximum,chor2000multiple, rogers1999multiple} are derived for binary models.
Theorem~$\ref{Theo1}$ complements those results in the sense that while the likelihood surfaces considered in those work may have multiple (or even a continuum of) local maxima, the stationary points of one\hyp dimensional likelihood functions are still unique.

This corollary also extends and clarifies a result from the first attempt to investigate the shape of the one\hyp dimensional phylogenetic likelihood functions \cite{fukami1989maximum}.
By studying the location of the solutions of phylogenetic likelihood functions, the paper proves that one\hyp dimensional phylogenetic likelihood functions have unique stationary points under the same model assumptions as Corollary $\ref{model}$.

This result also provides a full characterization of one\hyp dimensional likelihood functions of binary models (and those considered by Corollary~$\ref{model}$).
Indeed, since the derivatives of log-likelihood functions are continuous with at most one zero, this result implies that:
\begin{enumerate}
\item[1. ] If there is no stationary point, then $\ell(t)$ is a monotonic function (either strictly decreasing or strictly increasing).
\item[2. ] If the stationary point $t_0$ exists and is unique, then the function is increasing in the interval $(0, t_0)$ and is decreasing in $(t_0, \infty)$.
\end{enumerate}

This simplicity of the shapes of phylogenetic likelihood functions provides a strong theoretical foundation for the use of simple optimization methods to locate the maximum likelihood branch length.
However, we emphasize that these results are only about one\hyp dimensional phylogenetic likelihood functions and do not mean that there is a unique (multivariate) stationary point of the likelihood surface or that simple hill-climbing methods will find this optima.

\section{Algebraic structures on the space of all logarithmic relative frequency patterns under the Kimura 2-parameter model}
\label{sec:algebra}

While Section 3 provides a uniqueness result for the maximum likelihood branch lengths under three simple models, the result does not extend to more general models.
In fact, as we will illustrate in the next section, the shapes of likelihood functions under the Kimura 2-parameter model \cite{kimura1980simple} can be quite complicated, for example with multiple local and global maxima.

In order to enable theoretical analyses of phylogenetic likelihood functions under more complex evolutionary models, here we introduce the concept of conditional logarithmic frequency patterns and study the algebraic structures on the space of such patterns.

\begin{Definition}
 Given a rooted tree $\tau$ with root $\rho$ and $N$ taxa, some labelings $~\psi$=$(\psi_1, \ldots, \psi_S)\in \Omega^{N \times S}$ of its taxa and a vector of real constants $(c_1, \ldots, c_S)$ we define the \emph{logarithmic relative frequency pattern} $\phi(\tau, \psi, c)$ as the $r \times S$ matrix with entries
\[
\phi_{i,s} = c_s +  {\log  \sum_{a \in \mathcal{Z}_{i,s}}{ \pi(i) \prod_{(u,v) \in E(\tau)}{P^{uv}_{a_ua_v}(t_{uv})}}}
\]
for $i\in \Omega$, $s=1,\ldots,S$  and  $\mathcal{Z}_{i,s}$ being the set of all extensions $a$ of $\psi_s$ to all the nodes of $\tau$ such that $a(\rho)=i$.
\end{Definition}
For convenience, we will use the shorter term \emph{frequency pattern} to refer to a logarithmic relative frequency pattern.

In probabilistic terms, for a fixed site index $s$, the $(i,s)$-entry of a logarithmic relative frequency pattern $\phi(\tau, \psi, c)$ is (up to a constant $c_s$) the logarithm of the likelihood of observing state $i$ at the root of the tree, given leaf states $\psi_s$.
This definition is directly related to the formulation of the characteristic polynomials $\lambda_s$, whose coefficients $b_{ij}^s$ are the product of the probabilities of observing state $i$ and $j$ at the two end points of an edge, given that the labeling $\psi_s$ is observed at the taxa.
It is straightforward to verify that for models with uniform stationary distribution on a fixed tree, we have
\[
\log b_{ij}^s = \phi_{i,s} (\tau_1) + \phi_{j,s}(\tau_2) + \tilde c_s
\]
for all $i, j, s$, where $\tilde c_s$ is a constant depending only on $s$, and $\tau_1$ and $\tau_2$ are the trees obtained by removing the edge $e_0$ from the tree $\tau$ and rooting the newly created trees at the endpoints of $e_0$ (see the proof of Theorem 4.1 in the Appendix for more details).

Hence, to characterize the space of all phylogenetic characteristic polynomials under a given evolutionary model, we just need to characterize the space of all possible logarithmic relative frequency patterns under that model.

\begin{Definition}
We denote the space of all possible logarithmic relative frequency patterns under the Kimura 2-parameter model  by
\[
G = \{ \phi(\tau, \psi, c): \tau \in \mathcal{T}, \psi \in \Psi_{\tau}^S, c \in \mathbb{R}^S\}
\]
where $\mathcal{T}$ denotes the set of all rooted trees and $ \Psi_\tau^S$ denotes the set of all tuples $(\psi_1, \ldots, \psi_S)$ of $S$ labelings of the taxa of $\tau$.
\end{Definition}

The goal of this section is to establish that for any sequence of $S$ column vectors $v_1, v_2, \ldots, v_S$ in $\mathbb{R}^{4}$, there exists a tree $\tau$ under the Kimura 2-parameter model, labelings $\psi=(\psi_1, \psi_{2}, \ldots, \psi_S)$ of its taxa and a vector of real constants $c$ such that
\[
\phi(\tau, \psi,c) = [v_1~~ v_2~~\ldots~~v_S].
\]

The existence of such tree is guaranteed indirectly by proving that under the Kimura 2-parameter model:
\begin{itemize}
\item[ 1. ] $G$ is an algebraic subgroup of $(\mathbb{R}^{4\times S}, +)$.
\item[ 2. ] $G$ is a linear subspace of $\mathbb{R}^{4\times S}$.
\item [ 3. ] $G$ is equal to $\mathbb{R}^{4\times S}$ itself.
\end{itemize}

Noting again that the stationary distribution of the Kimura 2-parameter model is the uniform distribution across states $\pi = (1/4,1/4,1/4,1/4)$, the first two steps are confirmed by the following theorem.
\begin{Theorem}

If the stationary frequency of the evolutionary model is the same for every state, then the following properties hold:
\begin{itemize}
\item[1. ] $(G,+)$ is a subgroup of $(\mathbb{R}^{4\times S}, +)$.
\item[2. ] $G$ is path-connected.
\item[3. ] $G$ is a linear subspace of $\mathbb{R}^{4\times S}$.
\end{itemize}
\label{group}
\end{Theorem}

\begin{figure}
\centering
  \includegraphics[width=0.8\linewidth]{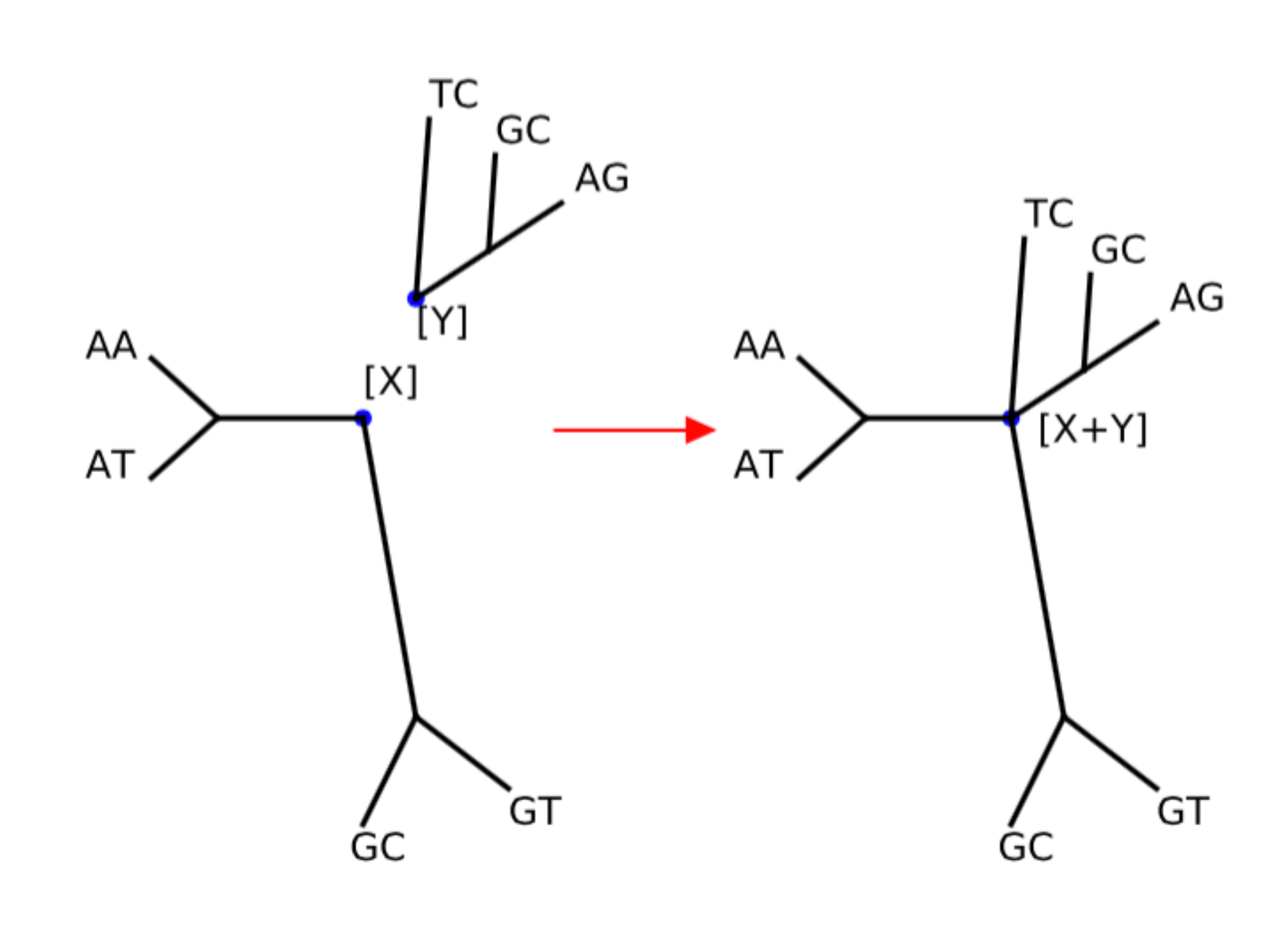}
  \caption{$G$ is closed under addition: we can add two frequency patterns $[X]$ and $[Y]$ just by gluing the roots of the two corresponding trees, labeling the taxa of $\tau$ correspondingly and taking the pattern at the new root.}
  \label{addition}
\end{figure}

\begin{figure}
\centering
  \includegraphics[width=0.6\linewidth]{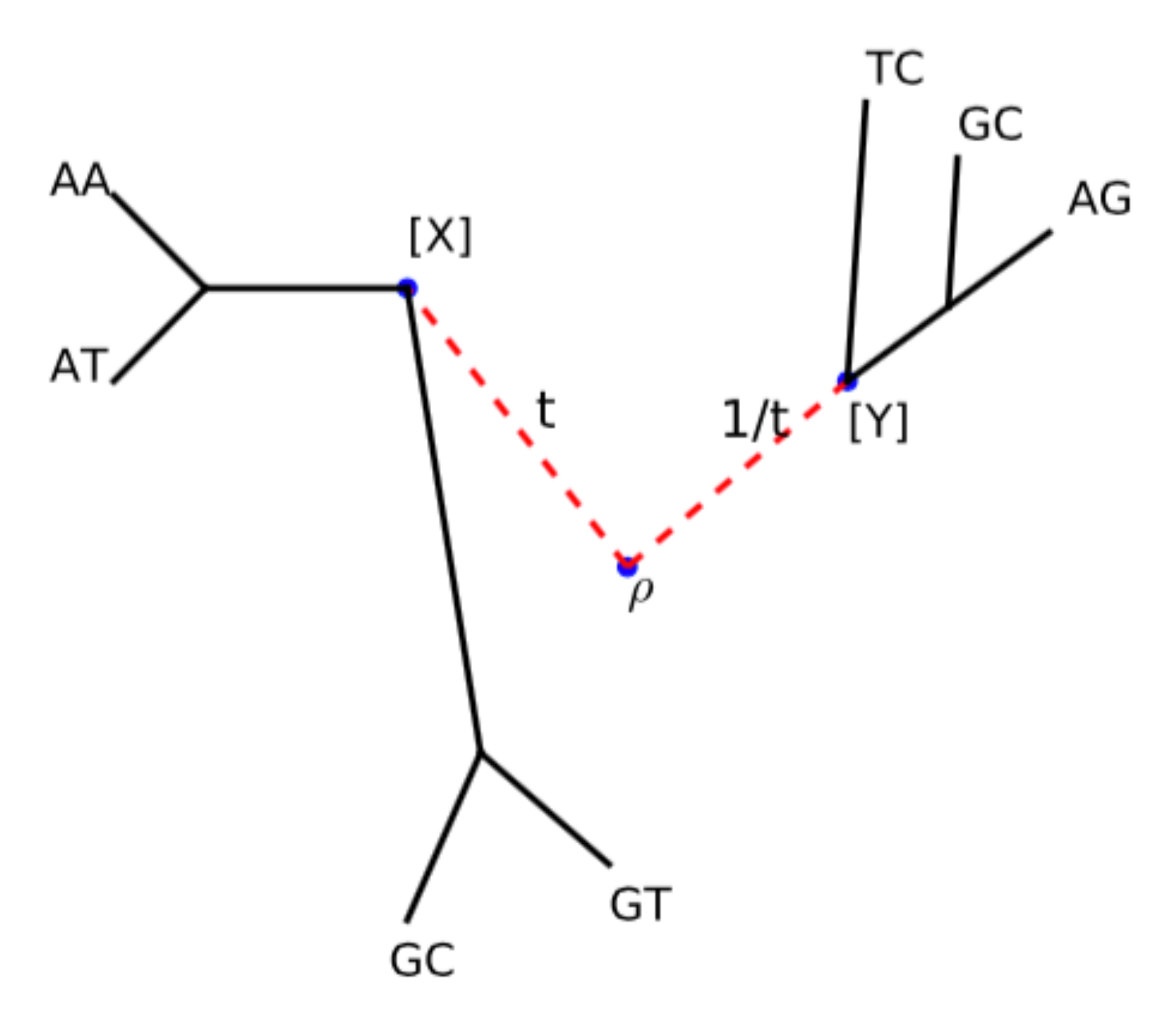}
  \caption{$G$ is path-connected: We can connect any two patterns $[X]$ and $[Y]$ by adding a new root $\rho$, joining it with the roots of the two corresponding trees with two new edges of length $t$ and $1/t$, respectively, and making $\rho$ the root of $\tau$.}
 \label{connect}
\end{figure}

\begin{proof}[Sketch of proof]
A detailed proof of this Theorem is provided in the Appendix, but the main arguments can be simply illustrated.
The fact that $G$ is closed under addition follows because we can add two frequency patterns just by gluing the roots of the two corresponding trees, labeling the taxa of $\tau$ correspondingly and taking the pattern at the new root (Figure $\ref{addition}$).
Similarly, we can create the inverse of a pattern by gluing all permuted versions of its corresponding tree (with an appropriate vector of real constants).

To prove that $G$ is path-connected, given two arbitrary trees with roots $\rho_1, \rho_2$, we create a new tree  by adding a new root $\rho$, joining $\rho_1, \rho_2$ with $\rho$ by two new edges of length $t$ and $1/t$, respectively, and making $\rho$ the root of $\tau$ (Figure $\ref{connect}$).
By varying $t$ continuously from zero to infinity, we can make a continuous path in $G$ that connects the two frequency patterns.
Since any path-connected subgroup of $\mathbb{R}^n$ is a linear subspace \cite{hayashida1949arc}, so is $G$.

\end{proof}

We note that although the aforementioned arguments are made for the Kimura 2-parameter model, which describes a model of DNA evolution ($r=4$), Theorem~$\ref{group}$ only requires that the stationary frequency of the evolutionary model is the same for every state.
Hence, this result also extends to models with more parameters.

Similarly, the fact that $(G,+)$ is a subgroup of $(\mathbb{R}^{r\times S}, +)$ can be established under the assumption that the root distribution $\pi$ is uniform, without assuming that it is the stationary distribution of the evolutionary process.
However, our current approach requires the uniform root distribution to be the stationary distribution for the proof of path-connectivity of $G$, and an alternative approach to the proof of path-connectivity will be needed if we want to extend the analyses to a more general framework.

Recalling that the Kimura 2-parameter model corresponds to the uniform stationary distribution and a family of rate matrices $Q_{\kappa}$ indexed by $\kappa$, the transition/transversion rate ratio, we then establish that when $\kappa=3$, the space of all frequency pattern $G=\mathbb{R}^{4\times S}$.
The proof is done through proving by induction that $G$ contains $4\times S$ independent frequency patterns (also proven in the Appendix):

\begin{Theorem}
The set of all possible logarithmic conditional frequency patterns with $S$ sites under the Kimura 2-parameter model with $\kappa=3$ is equal to $\mathbb{R}^{4\times S}$.
\label{Kimura}
\end{Theorem}

With those results, we finally can establish the main theorems of the section.

\begin{Theorem}
 For any sequence of column vectors $v_1, v_2, \ldots, v_S$ in $\mathbb{R}^{4}$, there exists a rooted tree $\tau$ under the Kimura 2-parameter model with $\kappa=3$, $S$ labelings $\psi_1, \psi_{2}, \ldots, \psi_S $ of its taxa, and a vector of real constants $c$ such that
 \[
\phi(\tau, \psi,c) = [v_1~~ v_2 \ldots v_S].
\]
\label{main}
\end{Theorem}
While Theorem~$\ref{main}$ provides a theoretical guarantee about the existence of a tree under the Kimura model with a given frequency patterns, the proof is not constructive.
This raises some concerns about the practicality of the approach.
For example, one can not derive an estimation of the number of edges required to produce a given frequency pattern.
Those concerns are addressed by the following theorem.

\begin{Theorem}
A tree as in Theorem~$\ref{main}$ can be constructed with at most $64S$ edges.
\label{main2}
\end{Theorem}

Not only does the theorem provide an upper bound on the number of edges required to construct a tree with a given frequency pattern, its proof also provides a simple algorithm to construct such a tree.

\begin{proof}[Proof of Theorem~$\ref{main2}$]

The main steps of the proof are as follows:

\begin{itemize}
\item[Step 1. ] As shown in the Appendix, any frequency pattern of the form $[x, 0, 0, 0 ]^t$ can be produced (up to a real constant $c_1$) by a tree $\tau$ with 4 edges and some labeling $\psi$ of its taxa.

\item[Step 2. ] Using $\tau$ from Step 1, we create a tree $\tau'$ of 16 edges by gluing the roots of 4 different versions $\tau_1, \tau_2, \tau_3, \tau_4$ of $\tau$ together and define $S$ labelings of $\tau'$ as follows.

\begin{itemize}
\item For $s=1$, we copy the labeling of $\tau$ onto $\tau'$.
\[
\psi_1(a) = \psi(a)
\]
for each taxon $a$ of $\tau_1, \tau_2, \tau_3, \tau_4$.
\item For all $s\ge 2$, the labelings are defined as follows:
\[
\psi_s (a) = \sigma^j(\psi (a)) \h \text{if $a$ is a taxon of $\tau_j$}
\]
where $\sigma$ is the permutation $(A ~ G ~ T ~ C)$ in cycle notation.
\end{itemize}

The construction of $\tau'$ is similar to the construction of the inverse of elements in the group $G$ in the proof of Theorem~$\ref{group}$.
Because of symmetry, for $s \ge 2$, the frequency pattern corresponding to site $s$ at the root of the newly created tree will be the same for every state while for $s = 1$, the frequency pattern of $\tau'$ is  obtained by multiplying the frequency pattern of $\tau$ by a factor of 4.

We deduce that the pattern created by $(\tau', \{\psi_i\} )$ is:
\[ \left( \begin{array}{cccc}
4x & 0 & \ldots & 0 \\
0 & 0 & \ldots & 0 \\
0 & 0 & \ldots & 0 \\
0 & 0 & \ldots & 0 \end{array} \right) + \left( \begin{array}{cccc}
c_1 & c_2 & \ldots & c_S \\
c_1 & c_2 & \ldots & c_S \\
c_1 & c_2 & \ldots & c_S \\
c_1 & c_2 & \ldots & c_S \end{array} \right)  \]
for some real constants $c_1, c_2, \ldots, c_S$.
\item[Step 3. ] By similar arguments, for any $i=1,2,3,4$ and $s=1,2, \ldots, S$, we can construct a tree of 16 edges for any patterns with $S$ sites whose only non-zero entry is at the $(i,s)$-position. Hence, it takes $16 \times 4S = 64S$ edges to construct a tree with an arbitrary given frequency pattern.

\end{itemize}

\end{proof}
\section{Non-uniqueness of stationary points: Kimura 2-parameter model}

In this section, we provide an example for which there are multiple stationary points of the likelihood function.
To construct such an example, we find two polynomials $p_1(x)$ and $p_2(x)$ with coefficients $b_1$, $b_2$  such that the product $p_1p_2$ has 2 local maxima in $[0,1]$, and
$p_1$ and $p_2$ can be expressed as positive linear combination of the basis polynomial functions $P_i$ derived from an evolutionary model (as will be carefully described in this section).
This gives a counter-example with $S=2$ sites.

Consider the Kimura 2-parameter model with $\kappa = 3$ which has the rate matrix
\begin{equation}
Q=\left( \begin{array}{cccc}
-5/8 & 3/8 & 1/8 & 1/8 \\
3/8 & -5/8 & 1/8 & 1/8 \\
1/8 & 1/8 & -5/8 & 3/8 \\
1/8 & 1/8 & 3/8 & -5/8 \end{array} \right).
\label{qkimura}
\end{equation}
This matrix has eigenvalues $0>-\gamma > -2\gamma$ where $\gamma=0.5$.
The transition probabilities under this evolutionary model can be computed explicitly by
\begin{align*}
\h P_1(t) &= 0.25+0.25\exp(-0.5t) +  0.5\exp(-t) \\
P_2(t) &= 0.25 + 0.25\exp(-0.5t) - 0.5\exp(-t) \\
P_3(t) &= P_4(t) = 0.25 -0.25\exp(-0.5t)
\end{align*}
where $P_1(t), P_2(t), P_3(t), P_4(t)$ are the probabilities of transitioning from state $A$ to state $A, T, G, C$, respectively.
This simple model is ``universal'' in an appropriate sense as shown in the end of the paper.

This leads to a representation of the likelihood as the product of two different linear combinations of the transition polynomials
\begin{align}
\h P_1(x) &= 0.25+0.25x +  0.5x^2 \nonumber \\
P_2(x) &= 0.25 + 0.25x - 0.5x^2 \label{equx} \\
P_3(x) &= P_4(x) = 0.25 -0.25x \nonumber
\end{align}
where $x=\exp(-0.5t)$.

We assume that the likelihood is computed by observing two sites $s_1$ and $s_2$, and that the edge of interest $e$ is a pendant edge with the observed values at that tip being $A$ for both sites.
Assume further that the state observation probabilities at the inner node of the edge $e$ are provided by
\[
b^1 = [0.24977275,  0.34067358,  0.2051904,  0.20436327 ]
\]
and
\[
 b^2 = [0.25,  0.16087344,  0.29328435,  0.29584221].
\]
As discussed earlier, the log-likelihood function can be computed as
\begin{equation}
\ell(t)=\log(\lambda_1(t))+ \log( \lambda_2(t))
\label{lambda}
\end{equation}
where
\[
\lambda_s(t) = \sum_{i =1}^4{b^s(i) P_i(t)}.
\]

Plots of the log-likelihood function $\ell$ and its perturbations (by varying the coefficients slightly) in terms of $x$ and $t$ are provided in Figure~$\ref{example}$ and Figure~$\ref{length}$, respectively.
The figures show that $\ell$ has three stationary points (two local maxima at $t_1 < t_2$ and one local minimum), all in the interval $[0,1]$.
The fact that $\ell(t_1)> \ell(t_2)$ for some cases and $\ell(t_1)<\ell(t_2)$ for some others indicates that there exist some values of $b_i^s$ such that $\ell(t_1)=\ell(t_2)$, i.e.
 the smoothly varying likelihood function can even have two global maxima.

\begin{figure}
\centering
  \includegraphics[width=0.7\linewidth]{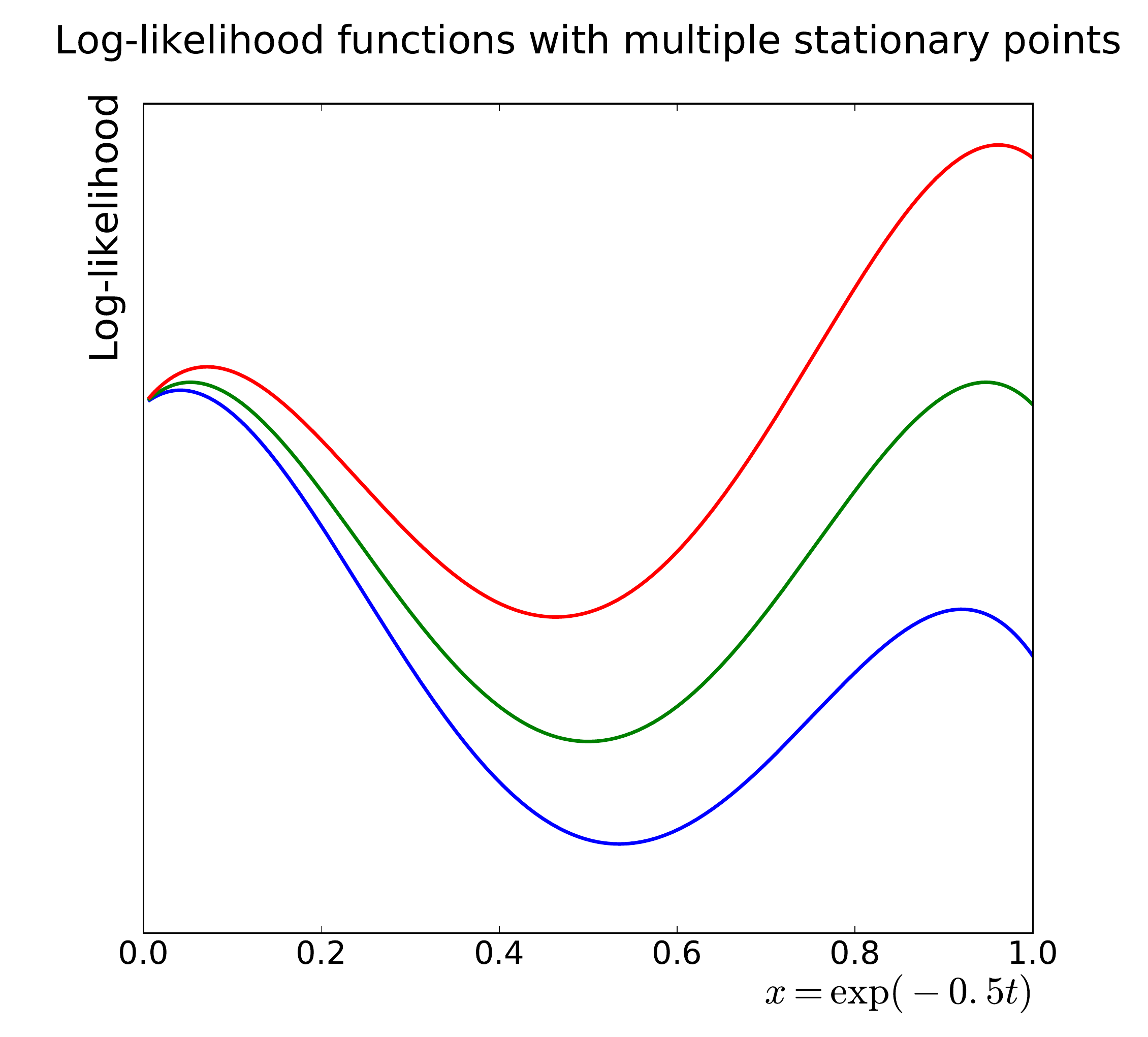}
  \caption{The log-likelihood \eqref{lambda} as a function of $x=\exp{(-0.5t)}$ for various values of the coefficients of the characteristic polynomial. }
  \label{example}
\end{figure}

We note that these examples can be achieved under the assumption that given any positive coefficients $b_i^s$ of the inner node, we can find some trees under the Kimura 2-parameter model with these precise coefficients.
This assumption is confirmed by the following result, proven in the Appendix.

\begin{Theorem}
For every set of positive coefficients $\eta_i^s$, there exist a phylogenetic tree $\tau$ and $S$ labelings $\psi_1, \psi_{2}, \ldots, \psi_S $ of the taxa such that for some edge $e$ in $\tau$, the one\hyp dimensional likelihood function on $e$ under the Kimura 2-parameter model with $\kappa=3$ satisfies
\[
\ell(\tau,t) = C_0 + \sum_{s=1}^S{\log{\left( \sum_{i}{\eta^s_{i} P_{i}(t)}
\right)}}
\]
where $P_i(t)$ is the probability of transition from state $A$ to state $i$ and $C_0$ is a constant. Moreover, such a tree $\tau$ can be constructed with at most $64S+1$ edges.

\label{density}
\end{Theorem}

In our examples, the upper bound on the number of edges to produce the given frequency pattern is $64 \times 2+1 = 129$ edges.

\begin{Remark}
While the algorithm to construct a tree given a frequency pattern given by Theorem~$\ref{main2}$ always outputs a star-tree (a tree without internal edges), we note that
\begin{enumerate}
\item[1. ] We can approximate any star tree by resolved trees with arbitrary precision.
\item[2. ] The maximum number of stationary points of a polynomial of degree four is 3, hence small perturbations on the coefficient of a polynomial of degree four with three stationary points do not change the number of stationary points.
\end{enumerate}

We deduce that there are resolved trees for which the one\hyp dimensional likelihood function on certain edges have multiple maxima.
\end{Remark}

Since a resolved tree with $n$ taxa has $2n-3$ edges, the upper bound on the number of edges of a resolved tree for which the one\hyp dimensional likelihood function on certain edges has multiple maxima is $2 \times 129 - 3 = 255$ edges.

\begin{figure}
\centering
  \includegraphics[width=0.7\linewidth]{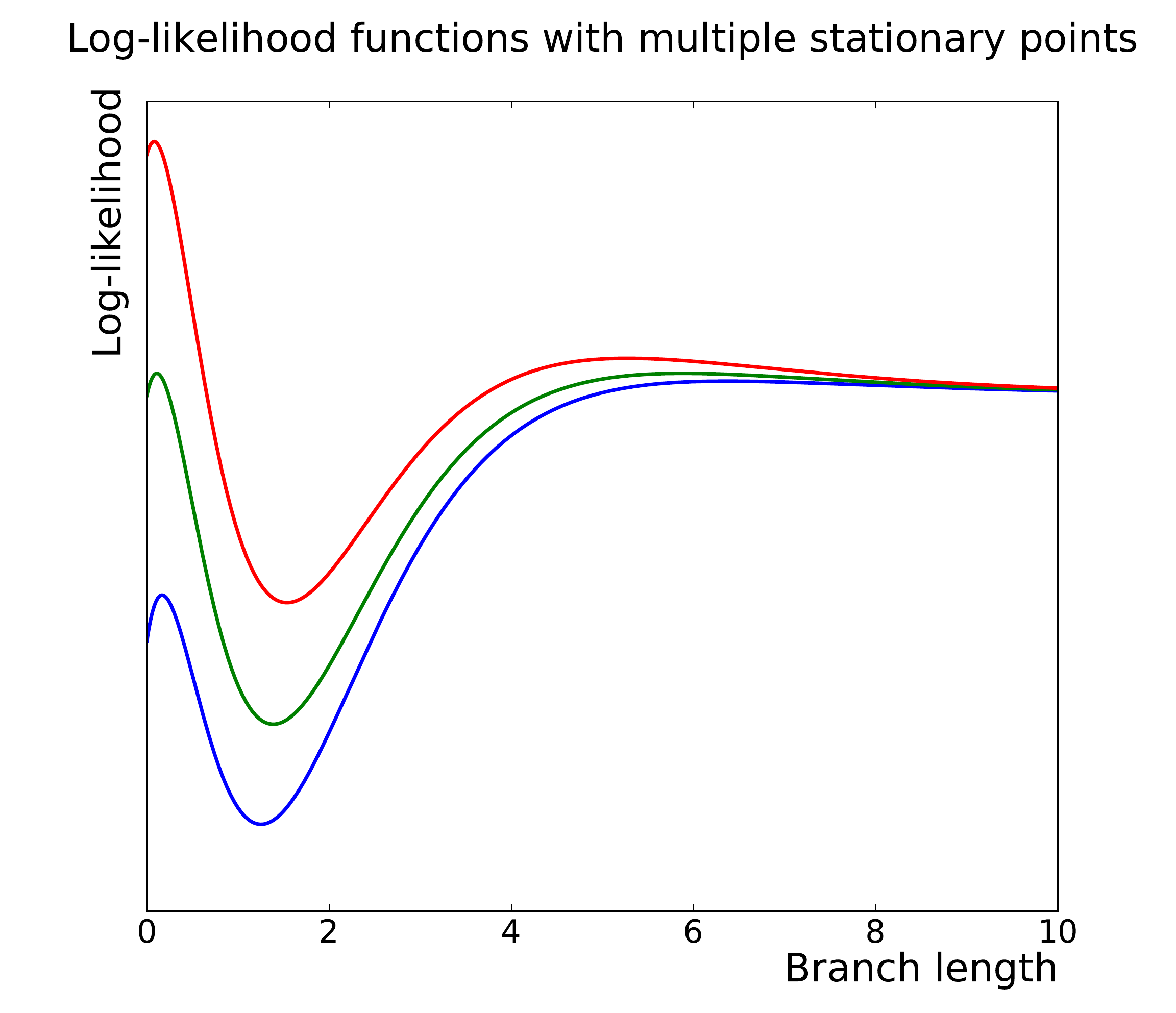}
  \caption{The log-likelihood \eqref{lambda} as a function of branch length $t$ for various values of the coefficients of the characteristic polynomial.  }
  \label{length}
\end{figure}

\section{Universality and complexity of the Kimura 2-parameter model}
\label{sec:6}
As we discussed earlier in the paper, the main idea behind the results in Section 3 and Section 4 is that by using the Fundamental Theorem of Algebra, we can decompose a complicated evolutionary model into smaller modules, each of which can be approximated either by a ``linear'' model or by a ``quadratic'' model.
This paradigm focuses on the branch lengths of the tree and is independent of the state space $\Omega$ of the evolutionary model, which provides a way to represent advanced evolutionary models (amino-acid models, codon models) by simple ones (nucleotide models).

This motivates the problem of constructing a complete characterization of one\hyp dimensional likelihood functions.
The main question is: does there exist an evolutionary model that can represent all one\hyp dimensional likelihood functions of any time-reversible evolutionary model?

Such a model $\mathcal{M}$, if it exists, and which we will refer to as a universal model, needs to satisfy the following two conditions:

\begin{enumerate}
\item[1. ] All one\hyp dimensional likelihood functions under any reversible evolutionary model can be written as a product of polynomials, each of which is a positive linear combination of the transition polynomials of $\mathcal{M}$.
\item[2. ] For every set of positive coefficients $b_{ij}^s$, there exists a phylogenetic tree $\tau$ and $S$ labelings $\psi_1, \psi_{2}, \ldots, \psi_S $ of the taxa such that for some edge $e$ in $\tau$, the one\hyp dimensional likelihood function on $e$ under the $\mathcal{M}$ satisfies
\[
\ell(\tau,t) = C_1 + \sum_{s=1}^S{\log{\left( \sum_{ij}{b^s_{ij} P_{ij}(t)}
\right)}}
\]
for some constant $C_1$.
\end{enumerate}

In this section, we will prove that the Kimura 2-parameter model with $\kappa = 3$ is, in fact, a universal model.
The key components of the proof are Theorem~$\ref{density}$, the Fundamental Theorem of Algebra and the fact that the transition polynomials of the Kimura 2-parameter model effectively span a large class of linear and quadratic polynomials.

\subsection{Universality of the Kimura 2-parameter model}

We first make the following observation, proven in the Appendix.

\begin{Lemma}
If $f$ is a real-coefficient polynomial that satisfies
\begin{enumerate}
\item[1. ] $f$ is positive on $[0,1]$,
\item[2. ] $\text{deg} ~f =1$ or $f$ is a quadratic polynomial with no real root,
\end{enumerate}
then $f$ can be written as positive linear combination of the transition polynomials of the Kimura 2-parameter model if and only if
\begin{enumerate}
\item[1. ] $\text{deg} ~f =1$ and $f(-1) > 0$,

or

\item[2. ] $\text{deg} ~f =2$ and $f$ has no root inside the set
\begin{equation}
B=\{z \in \mathbb{C}: |z+1|\le1 ~~ \text{or} ~~ |z-1|\le\sqrt{2}\}.
\label{star}
\end{equation}
\end{enumerate}
\label{lem1}
\end{Lemma}

This enables us to establish the universality of the Kimura 2-parameter model.

\begin{Theorem}[Universality]
If $L$ is a one\hyp dimensional phylogenetic likelihood function of a tree under an arbitrary time-reversible model that satisfies Assumption~\ref{as1}, then up to translation and rescaling, $L$ is equal to a one\hyp dimensional likelihood under the Kimura 2-parameter model.

That is, there exist $c_1, c_2, c_3 >0$ such that
\[
L(t) = c_2 L_{\operatorname{K2P}}(\tau, \psi, c_3 t) - c_1, \h \forall t \in [0,\infty),
\]
where $L_{\operatorname{K2P}}(\tau, \psi, \cdot)$ is the one\hyp dimensional likelihood function under the Kimura 2-parameter model on some edge of a tree $\tau$ with labeling $\psi$.
\label{uni}
\end{Theorem}

\begin{proof}
Assumption~\ref{as1} implies that the function
\[
\mathcal{L}(x) := L \left(-\frac{1}{\gamma} \log x \right)
\]
is a polynomial in $x$ for some $\gamma>0$. Since $\mathcal{L}$ is continuous and the set $B$ defined by $\eqref{star}$ is compact, if we define
\[
c_1 = 1 + \sup_{z \in B}{|\mathcal{L}(z)|},
\]
then by the triangle inequality, the polynomial $\mathcal{L}(x)+ c_1$ has no root in $B$.

By the Fundamental Theorem of Algebra, the polynomial $\mathcal{L}(x)+ c_1$ can be written as
\[
\mathcal{L}(x)+ c_1 =  \prod_{s=1}^{S}{g_s(x)},
\]
where each $g_s$ is either a quadratic polynomial with no real root, or a polynomial of degree 1.
Moreover, each $g_s$ is positive on $[0,1]$ and has no root in $B$ (which also implies $g_s(-1)>0$ if $\text{deg}~ g_s=1$).
Lemma~$\ref{lem1}$ implies that each $g_s$ can be written as a positive linear combination of the transition polynomials of the Kimura 2-parameter model
\[
g_s(x) = \sum_{ij} {b^s_{ij} P_{ij}(x)}.
\]

We deduce that
\[
\log(\mathcal{L}(x)+c_1) = \sum_{s=1}^S {\log \left( \sum_{ij} {b^s_{ij} P_{ij}(x)}\right)}.
\]

We recall that the Kimura 2-parameter model has symmetries such that any transition probability $P_{ij}(t)$ is in fact equal to $P_{Al}(t)=P_l(t)$ for some $l$.
Therefore, by grouping
\[
\eta^s_l := \sum_{i,j : P_{ij}=P_{Al}}{b^s_{ij}},
\]
we have
\[
\log(\mathcal{L}(x)+c_1)=\sum_{s=1}^S {\log \left( \sum_{l} {\eta^s_{l} P_{l}(x)}\right)}.
\]
Also, the characteristic polynomial for the Kimura 2-parameter model \eqref{qkimura} with $\kappa=3$ is parameterized by $x=\exp(-0.5 t)$ such that the one\hyp dimensional likelihood $L_{\operatorname{K2P}}(\tau, \psi, t)$ satisfies
\[
L_{\operatorname{K2P}}(\tau, \psi, t) = \mathcal{L}_{\operatorname{K2P}}(\tau, \psi, \exp(-0.5 t)).
\]

Now, Theorem~$\ref{density}$ guarantees that there exists a tuple $(\tau, \psi)$ under the Kimura 2-parameter model on an edge of the tree such that
\[
\log \mathcal{L}_{\operatorname{K2P}}(\tau, \psi, x) = - \log c_2 + \sum_{s=1}^S {\log \left( \sum_{l} {\eta^s_{l} P_{l}(x)}\right)}
\]
for some positive constant $c_2$.

In other words, we have
\[
\mathcal{L}(x) = c_2 \mathcal{L}_{\operatorname{K2P}}(\tau, \psi, x) - c_1,  \h \forall x \in (0,1].
\]

Hence,
\[
L \left(-\frac{1}{\gamma} \log x \right) = c_2 L_{\operatorname{K2P}}(\tau, \psi, -2 \log x) - c_1, \h \forall x \in (0,1],
\]
or
\[
L(t) = c_2 L_{\operatorname{K2P}}(\tau, \psi, c_3 t) - c_1, \h c_3 = \gamma/2, \h \forall t \in [0,\infty).
\]
That is, up to translation and rescaling, $L$ is equal to a one\hyp dimensional phylogenetic likelihood function under the Kimura 2-parameter model.
\end{proof}

Since the set of rate matrices for a given evolutionary model that satisfy Assumption~$\ref{as1}$ is dense in the set of all possible rate matrices under the same evolutionary model (Remark~$\ref{rem1}$), we also have the following corollary.
\begin{Corollary}
Any one\hyp dimensional phylogenetic likelihood function under an arbitrary time-reversible evolutionary model can be uniformly approximated with arbitrary precision by (rescaled and translated) one\hyp dimensional phylogenetic likelihood functions under the Kimura 2-parameter model.
\label{col1}
\end{Corollary}

We also note that the rescaling and translation constants in the statements of Theorem~$\ref{uni}$ can not be removed: Lemma~$\ref{lem1}$ indicates that some polynomial function can not be represented exactly as a Kimura 2-parameter likelihood function.
For example, one of the transition polynomials of the Jukes-Cantor model is
\[
J(x) = 0.25 + 0.75 x
\]
which has $J(-1)<0$.
For this reason, some likelihood functions under the Jukes-Cantor model may not be represented exactly by the Kimura 2-parameter model without adjusting by an additive constant.

\subsection{Complexity of the Kimura 2-parameter model}

The universality results in the previous section can be adapted easily to analyze the set of all one\hyp dimensional phylogenetic likelihood functions under the Kimura 2-parameter model.
The following complexity results imply that one\hyp dimensional likelihood functions under advanced evolutionary models can be more complex than it is typically assumed by phylogenetic inference algorithms.

First, it is straightforward to check that Theorem~\ref{uni} still holds (without changing the proof) if we replace the one\hyp dimensional phylogenetic likelihood function $L$ with an arbitrary polynomial $P$ in $x = \exp(-\gamma t)$ for some $\gamma>0$ and relax Assumption~$\ref{as1}$.
Moreover, if $P$ is of degree $n$, then by Theorem~$\ref{density}$, it can be represented by a one\hyp dimensional likelihood function of a tree with at most $(64n+1)$ edges with respect to some $n$-site labeling of its taxa.

\begin{Corollary}
Given an arbitrary polynomial $P$ of degree $n$ and $\gamma>0$, then up to translation and rescaling, $P(\exp(-\gamma t))$ is equal to a one\hyp dimensional likelihood under the Kimura 2-parameter model on a phylogeny with at most $64n +1$ edges.
\label{col2}
\end{Corollary}
This corollary indicates that by increasing the number of sites and the size of the tree, we can obtain likelihood functions shaped like an arbitrary polynomial in the interval $[0,1]$. For example, given an arbitrary finite sequence $t_1, t_2, \ldots, t_k \in (0, \infty)$, we can construct a polynomial $P_k$ that peaks precisely at $x_k = \exp(-0.5t_k)$ and use Corollary~$\ref{col2}$ to obtain the following result.

\begin{Corollary}
Given an arbitrary finite sequence $t_1, t_2, \ldots,t_k \in (0,\infty)$, there exists a phylogenetic tree $\tau$ and some labeling of its taxa such that for some edge of the tree, the one\hyp dimensional likelihood function under the Kimura 2-parameter model peaks precisely at $t_1, t_2, \ldots, t_k $.

Furthermore, since rescaling and translation do not change the relative order of the likelihood values at the stationary points, we can make any of the $t_i$'s (or all of them) the function's global maxima.
\end{Corollary}

Finally, we can replace the phylogenetic likelihood functions in Corollary~$\ref{col1}$ by an arbitrary continuous function $f$ with finite limit to obtain the following density result.
\begin{Corollary}
The space of all rescaled and translated one\hyp dimensional phylogenetic likelihood functions under the Kimura 2-parameter model is dense in the space of all non-negative continuous functions on $[0, \infty)$ with finite limits.
\label{col3}
\end{Corollary}

\begin{proof}
Let $f$ be a continuous function on $[0, \infty)$ with finite limit. Define
\[
g(x)=f(-\log(x)) \h \forall x \in (0,1],
\]
then $g(x)$ can be extended continuously to $[0,1]$.
By Weierstrass's theorem \cite{farouki2012bernstein}, there exists a sequence of  positive polynomials $\{P_n\}$ such that
\[
\sup_{x \in [0,1]}{|P_n(x)-g(x)|} \to 0.
\]
This implies that
\[
 \sup_{t \in [0,\infty)}{|P_n(\exp(-t))-f(t)|} \to 0.
\]
On the other hand, we deduce from Corollary~$\ref{col1}$ that $P_n(\exp(-t))$ is, up to rescaling and translation, a one\hyp dimensional likelihood under the Kimura 2-parameter model.
This completes the proof.
\end{proof}

\section{Non-reversible Markov models of evolution}
\label{graphical}

As we mentioned in Section $\ref{background}$, the analyses in the previous sections can be described in the more general framework of probabilistic inference for graphical models.
In this framework, the likelihood function can be be defined as the marginal distribution on the leaf nodes of a joint probability distribution that factorizes over the edges of the tree via the non-negative kernels (also referred to as potential functions) $k_{e}(i, j, t)$.
The one\hyp dimensional phylogenetic likelihood functions can be obtained by fixing all but one branch length.

In this section, we briefly analyze the extent to which our analyses of one\hyp dimensional likelihood functions are valid in this more general setting.
As we illustrate below, the results in this section do not assume the reversibility of the kernels and thus apply for non-reversible models of evolution.
However, we need to modify our assumptions accordingly.

Several parts of our analysis rely on the core assumption that the kernel functions need to be polynomials of $x=\exp(-\gamma t)$ for some $\gamma>0$.
Thus we require the following assumption, which is the equivalent of Assumption~$\ref{as1}$ but in a more general setting.
\begin{Assumption}[Polynomial representation]
There exists a constant $\gamma_e>0$ and polynomials $p^{ij}_e(x)$ such that
\[
k_{e}(i, j, t) = p^{ij}_e(\exp(-\gamma_e t))\h \forall t,
\]
for all $i, j \in \Omega$ and $e \in E(T)$.
\label{as3}
\end{Assumption}
This assumption implies that the limit of $k_{e}(i, j, t)$ for large $t$ exists, that is, the kernels are stationary.
We also note that using the density results for Bernstein's polynomial approximation (see, for example, \cite{farouki2012bernstein}), any non-negative continuous function on $[0, \infty)$ with finite limit can be approximated with arbitrary precision by some kernels that satisfy Assumption~$\ref{as3}$.

Under this assumption, the characteristic polynomials $\lambda_s(x)$ can be defined in a similar manner and the result in Section $\ref{sec:unique}$ (Theorem $\ref{Theo1}$) is still valid.
\begin{Theorem}
Under Assumption $\ref{as3}$, if for every site index $s$, the polynomial $\lambda_s$ has only real roots, then one\hyp dimensional likelihood function has at most one stationary point.
Moreover, if such a point exists, it is also a global maximum.
\label{Theo1-general}
\end{Theorem}

While Section $\ref{sec:algebra}$ is specifically developed to analyze the Kimura 2-parameter model, the logarithmic relative frequency pattern can be extended easily by replacing the transition probabilities across the edge $e$ by the kernel $k_e$ and by setting the distribution of the root by the uniform distribution.
Building upon this concept, we can study the algebraic structure of the space of all frequency patterns and obtain a partial extension of Theorem $\ref{group}$ in a more general setting as described below.
However, the proofs of Theorems 4.2 and 4.3 are tailor-made for the Kimura 2-parameter model and are not easily extended to the general case.
We leave their extension as open problems.

We do obtain the following partial extension of Theorem~\ref{group} in the more general setting.
\begin{Theorem}
The following properties hold:
\begin{itemize}
\item[1. ] $(G,+)$ is a subgroup of $(\mathbb{R}^{r\times S}, +)$.

\item[2. ]  Assume
\[
\lim_{t \to \infty}{k_{e}(i, j, t)} =p^{ij}_e(0) = \frac{1}{r} \h \text{and} \h \lim_{t \to 0}{k_{e}(i, j, t)} = p_e^{ij}(1)=\delta_{ij}
\]
for all $i, j \in \Omega$ and $e \in E(T)$, where $\delta_{ij} =1$ if $i=j$ and $\delta_{ij} =0$ otherwise.
Then $G$ is a linear subspace of $\mathbb{R}^{r\times S}$.
\end{itemize}
\label{group-general}
\end{Theorem}

We recall that the frequency patterns can be defined without the polynomial representation of the likelihood.
Thus, in Theorem $\ref{group-general}$, Assumption $\ref{as3}$ is not required.
We further note that for part (1) of the theorem to be valid, we do not need to assume that the kernels are stationary.
In fact, the only condition required is that the distribution at the root is uniform.
To provide a proof for path-connectivity of $G$, however, the conditions about the behavior of the kernels at $0$ and $\infty$ are necessary.

Polynomial representation of likelihood functions (Assumption $\ref{as3}$) is needed to extend the results of Section $\ref{sec:6}$.
By the same arguments as in the proof of Theorem $\ref{uni}$, we obtain an equivalent result.
\begin{Theorem}
If $L$ is a one\hyp dimensional phylogenetic likelihood function of a tree under a model whose kernels satisfy Assumption~\ref{as3}, then up to translation and rescaling, $L$ is equal to a one\hyp dimensional likelihood under the Kimura 2-parameter model.
\end{Theorem}

\section{Conclusions and discussion}

In this work, we investigate the problem of characterizing the shape of one\hyp dimensional phylogenetic likelihood functions.
Our results classify all evolutionary models into two categories:
\begin{enumerate}
\item[1. ] For binary, Jukes-Cantor and Felsenstein 1981 models: the one\hyp dimensional likelihood function has at most one stationary point.
\item[2. ] For Kimura 2-parameter model and more advanced evolutionary models: the shape of the one\hyp dimensional likelihood function can be much more complex.
In fact, the space of all rescaled and translated one\hyp dimensional phylogenetic likelihood functions under such a model is dense in the set of all non-negative continuous functions on $[0, \infty)$ with finite limits.
\end{enumerate}

Despite the complexity of the one\hyp dimensional likelihood functions under advanced evolutionary models, we prove that all one\hyp dimensional phylogenetic likelihood function are essentially Kimura 2-parameter likelihood functions.
This result establishes a strong foundation for the use of the Kimura 2-parameter as the building block of all evolutionary models.

Our results are based on two novel techniques.
First, we introduce and use \emph{characteristic polynomial representations} of one\hyp dimensional phylogenetic likelihood functions and the Fundamental Theorem of Algebra to decompose any evolutionary models into smaller modules, each of which resembles the Kimura 2-parameter model.
Second, we introduce the new concept of \emph{logarithmic relative frequency patterns} and analyze algebraic structures on the space of such patterns.
These structures open a new way to explore the space of all possible likelihood functions.
Moreover, by analyzing these structures, we are able to tackle the inverse problem of constructing a phylogenetic tree that has a given frequency pattern at the root.
This enables us to construct phylogenetic trees that approximate any given likelihood function with arbitrary precision.

There are several avenues for improvement.
Firstly, while we know that the shape of one\hyp dimensional likelihood function can be very complex, it is not clear how frequently multimodality might be encountered in practice and to which degree it affects the accuracy of phylogenetic algorithms.
Since the space of high degree polynomials are dominated by multimodal functions, one might expect that as the number of sites and the size of the tree increase, multimodality becomes more likely.
However, since the space of phylogenies is known to possess considerable hidden structure which sometimes lead to counter-intuitive properties, careful analysis of the space of all rescaled and translated one\hyp dimensional phylogenetic likelihood functions under the Kimura 2-parameter model are required to evaluate this hypothesis.
Secondly, although the focus of this work is on one\hyp dimensional phylogenetic likelihood functions, it is possible to utilize the framework we propose to study full phylogenetic likelihood functions.
This will be a subject for future work.

\section{Acknowledgements}
We are grateful to Connor McCoy and Brian Claywell for their work on surrogate functions for likelihood computation, which motivated this research.
This work is supported by DMS-1223057 and CISE-1564137 from the National Science Foundation and U54GM111274 from the National Institutes of Health.

\bibliographystyle{ieeetr}
\bibliography{biblio}

\section{Appendix}

\begin{proof}[Proof of Remark~$\ref{rem1}$]
If we denote the entries of the diagonalizing matrix $M$ and $N$ of $Q$ by $m_{ij}$ and $n_{ij}$, respectively, then \[
Q_{ij}= \sum_{k}{m_{ik} e^{-r_k} n_{kj}}.
\]
where $-r_k$ are the eigenvalues of $Q$. (The eigenvalues are known to be non-positive, so $r_k$ are non-negative.)

Since the set of rational numbers $\mathbb{Q}$ is dense in $\mathbb{R}^+$, we can find $r(k,l) \in \mathbb{Q}^+$ such that for all $k$, $r(k,l) \to r_k$ as $l$ approaches infinity. If we define
\[
Q^l_{ij}= \sum_{k}{m_{ik} e^{-r(k,l)} n_{kj}},
\]
then $Q^l \to Q$ element-wise as $l$ approaches infinity. Since $r(k,l)$ are all rational the matrices are of fixed finite dimension, we can also find $\gamma_l >0$ and $d(k,l) \in \mathbb{N}$ such that $r(k,l)=d(k,l) \gamma_l$.
\end{proof}

\begin{proof}[Proof of Theorem~$\ref{group}$]

We define the equivalence relation $\sim$ on $\mathbb{R}^{4 \times S}$ as follows: $u \sim v$ if and only if there exists a vector of real constants $c=(c_1, \ldots, c_S)$ such that for all $i =1,2,3,4$ and $s=1, \ldots, S$, we have
\[
u_{i,s} = v_{i,s} + c_s.
\]

If we define
\[
[h(\tau, \psi)]_{i,s} = \log  \sum_{a \in \mathcal{Z}_{i,s}}{ \prod_{(u,v) \in E(\tau)}{P^{uv}_{a_ua_v}(t_{uv})}}
\]
for $i\in \Omega$, $s=1,\ldots,S$  and  $\mathcal{Z}_{i,s}$ being the set of all extensions $a$ of $\psi_s$ to all the nodes of $\tau$ such that $a(\rho)=i$.

Recall that the Kimura 2-parameter model has a uniform stationary distribution, for all $\tau, \psi, c$, we have $\phi(\tau, \psi,c) \sim h(\tau, \psi)$.

\begin{enumerate}

\item[1.] (Addition): Consider any two elements $x_1, x_2 \in G$. By the definition of $G$ and since the stationary frequency of the evolutionary model is the same for every state, there exist trees $\tau_1, \tau_2$ with $n_1, n_2$ taxa and labelings $\psi_1, \psi_2$ such that
 \[
 x_i \sim h(\tau_i, \psi_i) , \h i=1,2.
 \]
If we construct a new tree $\tau$ from $\tau_1$ and $\tau_2$ by gluing the roots $\rho_1, \rho_2$ and label the taxa of $\tau$ corresponding to $\psi_1, \psi_2$, then we have
\[
[h(\tau, \psi)]_{i,s} = \log  \sum_{a \in \mathcal{Z}_{i,s}}{ \prod_{(u,v) \in E(\tau_1)}{P^{uv}_{a^1_ua^2_v}(t_{uv})}\prod_{(u,v)  \in E(\tau_2)}{P^{uv}_{a^2_ua^2_v}(t_{uv})}}
\]
where each term $a$ in the sum corresponds uniquely to a pair of extensions $(a^1,a^2)$ of $\psi_1^s, \psi_2^s$ to the internal nodes of $\tau_1, \tau_2$, respectively, such that $a^1(\rho)=a^2(\rho)=i$.

Therefore,
\begin{eqnarray*}
[h(\tau, \psi)]_{i,s}  &=& \log  \sum_{a^1}{ \prod_{(u,v)  \in E(\tau_1)}{P^{uv}_{a^1_ua^1_v}(t_{uv})}} +  \log  \sum_{a^2}{ \prod_{(u,v) \in E(\tau_2)}{P^{uv}_{a^2_ua^2_v}(t_{uv})}}\\
&=& [h(\tau_1, \psi_1)]_{i,s}  + [h(\tau_2, \psi_2)]_{i,s}
\end{eqnarray*}
for all $i \in \Omega$ and $s =1, \ldots, S$.

Therefore
\[
h(\tau, \psi) \sim h(\tau_1, \psi_1) + h(\tau_2, \psi_2)
\]
and
\[
x_1 + x_2 \sim h(\tau, \psi)  \in G
\]
which implies that $G$ is closed under addition.

\item[2.]

(Inverse): Consider any element $x \in G$ and its corresponding representative tree $\tau$ and labeling $\psi$. For any permutation $\sigma$ of the states, we define the labeling $\psi_\sigma$ as
\[
\psi_{\sigma}(\omega)=\sigma(\psi(\omega))
\]
for every taxon $\omega$ of $T$. For example, if $\sigma$ is the permutation $(A ~ G ~ T~ C)$ in cycle notation, then $\psi_{\sigma}$ is obtained from $\psi$ by replacing $A$ by $G$, $G$ by $T$, $T$ by $C$ and $C$ by $A$.

Now let $\sigma_0$ be a permutation of order $r$ on the state space $\Omega$, create $r$ identical copies $\tau_1, \tau_2, \ldots, \tau_r$ of the tree $\tau$ with labelings $\psi_{\sigma_0}$, $\psi_{\sigma_0^2}$, \ldots, $\psi_{\sigma_0^r}$ and glue the root of all the trees together with taxon labeling $\gamma$ corresponding to the labelings of $\tau_1, \tau_2, \ldots, \tau_r$.
Then because of symmetry, the frequency pattern $f$ at the root of the newly created tree $\mu$ will be the same for every state, i.e., $f \sim 0$. We deduce that $0 \in G$ and for every $x \in G$, there exists $y \in G$ such that $x+y=0$.

This property and the fact that $G$ is closed under addition prove that $(G,+)$ is a subgroup of $(\mathbb{R}^{r\times S}, +)$.

\item[3.]
(Connectedness): Consider any two elements $x_1, x_2 \in G$ and their corresponding trees $\tau_1, \tau_2$, labelings $\psi_1, \psi_2$ and vectors of real constants $c_1,c_2$. For any $\alpha \in (0,1)$, we create a new tree $\tau(\alpha)$ by adding a new root $\rho$, joining $\rho_1, \rho_2$ by new edges of length $t_1=\tan(\frac{\pi}{2}\alpha), t_2=1/t_1$, respectively.
We make $\rho$ the root of $\tau$ and label the taxa of $\tau$ according to  $\psi_1, \psi_2$.

Now we note that when $\alpha \to 0$, we have
\[
h(\tau(\alpha), \psi) \to h(\tau_1, \psi_1) + \log{\frac{1}{r}}
\]
since the contribution of $\tau_2$ becomes stationary (the stationary frequency is $1/r$ because of the model's symmetry). Similarly, when $\alpha \to 1$, we have
\[
h(\tau(\alpha), \psi) \to h(\tau_2, \psi_2) + \log{\frac{1}{r}}.
\]
Therefore, the function $g(\alpha)=\phi(\tau(\alpha), \psi)$ can be extended continuously to the closed interval $[0,1]$. By changing $c$ continuously from $c_1$ to $\log{(1/r)}$, varying $\alpha$ continuously from $0$ to $1$, then changing $c$ continuously from $\log{(1/r)}$ to $c_2$, we can make a path in $G$ that connects $x_1$ and $ x_2$.

\item[4.] Since any path-connected subgroup of $\mathbb{R}^n$ is a linear subspace \cite{hayashida1949arc}, so is $G$.

\end{enumerate}
\end{proof}

\begin{proof}[Proof of Theorem~$\ref{Kimura}$]

Denote by $\mathcal{H}$ the set of all rooted trees with one edge (which have varying branch lengths) and
\[
H = \{ \phi(\tau, \psi, c): \tau \in \mathcal{H}, \psi=(\psi_1, \psi_2, \ldots, \psi_S) \in \mathbb{R}^{S}, c \in \mathbb{R}^S\}.
\]
Note that in the context of this paper, trees with different branch lengths (or in other words, different values of $x$) are considered as different trees.
Thus the set $H$ defined here is non-trivial.

We have
\begin{equation}
[h(\tau,\psi)]_{j,s} = \log  P_{\psi_s j}(x)
\label{equind}
\end{equation}
where $x=\exp(-0.5t)$, $j =A, G, T, C$, and $t$ is the length of the unique edge of $\tau$.

Let $x_1=1/4$, $x_2=1/2$, $x_3=3/4$. We will prove, by induction on $S$, that $H$ contains $4\times S$ independent frequency patterns.

For $S=1$, by considering the 4 different patterns $(A), (G), (T), (C)$ at the only leaf and the 3 values of $x$ (corresponding to different branch lengths) described above, we can create a set of $4\times 3=12$ different pairs $(\tau, \psi)$.
A quick check by computer shows that the corresponding frequency patterns generated by those pairs span the whole vector space $\mathbb{R}^{4\times 1}$.
We can achieve similar result for $S=2$ with the patterns $(A,G), (G, T), (T,C), (C,A)$.

Now assume that for $S=n$, $H$ contains $4\times n$ independent frequency patterns of the form \eqref{equind}.

For $l=A, G, T, C$ and $x \in [0,1]$, we define the building blocks
\[ R_l(x):= \left[\log P_{lA }(x) ~~
\log P_{l G}(x) ~~
\log P_{l T}(x) ~~
\log P_{l C}(x) \right]
\]
\[ W_l(x):=\left( \begin{array}{c}
 R_l(x_1) \\
 R_l(x_2)  \\
 R_l(x_3)
\end{array} \right),\]

The induction hypothesis implies that there exist $4\times n$ independent frequency patterns of the form \eqref{equind}.
This means that for some labelings $\psi_1, \psi_2, \ldots, \psi_{4n}$, the block matrix
\[ J=\left( \begin{array}{cccc}
B_1 \\
B_2 \\
\cdots \\
B_{4n}
\end{array} \right)\]
has maximal rank $4n$, where
\[ B_s:=\left( \begin{array}{cccc}
R_{\psi_s^1}(x_1) \cdots  R_{\psi_s^n }(x_1) \\
R_{\psi_s^1 }(x_2) \cdots  R_{\psi_s^n }(x_2)  \\
R_{\psi_s^1 }(x_3) \cdots  R_{\psi_s^n }(x_3)
\end{array} \right).\]

 For $s=1, \ldots, 4n$, we consider all the labelings obtained by appending $\psi_s$ with one of the four nucleotides $A, G, T, C$. By doing so, we create a set of $48n$ different frequency patterns.
 We want to prove that the block matrix

\[ C =\left( \begin{array}{cccc}
B_1 \qquad W_{A}(x)\\
B_2 \hfill W_{A}(x) \\
\cdots \\
B_{4n} \hfill W_{A}(x)
\vspace{4pt}\\
B_1 \hfill W_{G}(x)\\
B_2 \hfill W_{G}(x) \\
\cdots \\
B_{4n} \hfill W_{G}(x)
\vspace{4pt}\\
B_1 \hfill W_{T}(x)\\
B_2 \hfill W_{T}(x) \\
\cdots \\
B_{4n} \hfill W_{T}(x)
\vspace{4pt}\\
B_1 \hfill W_{C}(x)\\
B_2 \hfill W_{C}(x) \\
\cdots \\
B_{4n} \hfill W_{C}(x)
\end{array} \right)\]
has maximal rank $4n +4$.

Note that this matrix is row-equivalent to
\[ \left( \begin{array}{cccc}
J \h U \\
0 \h V
\end{array} \right)\]
where each row of $V$ is of the form $R_{i}(x_k) - R_{A}(x_k)$ for $i= G, T, C$.
(This is done by subtracting the blocks $(B_s~R_{i}(x))$ by the block ($B_s ~ R_{A}(x)$) then rearranging the row to obtain the sub-matrix $J$ at the top-left corner.)

On the other hand, from the case $S=1$, we have
\[ \text{rank}\left( \begin{array}{cccc}
W_{A}(x)\\
W_{G}(x)\\
W_{T}(x)\\
W_{C}(x)\\
\end{array} \right)  = 4,\]
which implies that $\text{rank}(V)=4$. Hence, $\text{rank}(C)=\text{rank}(J)+\text{rank}(V)=4n+4$.

We deduce that for every $S$, the set $G$ of all possible logarithmic conditional frequency patterns with $S$ sites under the Kimura 2-parameter model is a linear subspace of $\mathbb{R}^{4\times S}$ (Theorem~\ref{group}) that contains $4S$ linearly independent vectors. This implies that $G=\mathbb{R}^{4\times S}$.

\end{proof}

\begin{proof}[Proof of Theorem~$\ref{main2}$ (Step 1)]

(Any pattern of the form $v=[x~ 0~ 0 ~0 ]$ can be produced by a tree $\tau$ with four edges.)

Denote
\[
x_1(t) = P_{AA}(t) \h x_2(t) = P_{AG}(t)
\]
\[
x_3(t) = P_{AT}(t) \h x_4(t) = P_{AC}(t)
\]
we note that in the Kimura 2-parameter model, $x_3(t)=x_4(t)$.

Now consider two trees $\tau_1$ and $\tau_2$, each with one edge, whose branch lengths are $t$ and $s$, respectively.
We label the only nodes of $\tau_1$ and $\tau_2$ by the patterns $\psi_1=(A)$ and $\psi_2=(G)$, and obtain the frequency patterns $f_1(t)$ and $f_2(s)$ respectively.
By gluing the roots of $\tau_1$ (1 edge) and the ``inverse'' of the tree $\tau_2$ (3 edges)), we obtain a tree $T(t,s)$ with 4 edges whose frequency pattern is equivalent to
\[
f_1(t) - f_2(s) \sim \left [\log \frac{x_1(t) x_4(s)}{x_4(t)x_2(s)}, \log \frac{x_2(t) x_4(s)}{ x_4(t)x_1(s)}, 0, 0 \right].
\]

On the other hand, we note that for the Kimura 2-parameter model \eqref{qkimura}, $$\frac{x_2(t)}{x_4(t)} = 1+ 2 \exp(-0.5t)$$ only admits values in the interval $[1, 3]$, while $x_1(s)/x_4(s)$ is a continuous decreasing function in $s$ that admits all values in the interval $[1, \infty)$. Hence, for every $t>0$, there exists a unique $k(t)>0$ such that
\[
 \frac{x_2(t) x_4(k(t))}{ x_4(t)x_1(k(t))} = 1.
\]
Moreover, $k(t)$ is a continuous function in $t$ and
\[
\lim_{t\to \infty}{k(t)} = \infty \h  \lim_{t\to 0}{k(t)} =k_0
\]
where $k_0$ satisfies $x_1(k_0)/x_4(k_0) = 3$.

Now if we denote
\[
g(t) = \frac{x_1(t) x_4(k(t))}{x_4(t)x_2(k(t))}
\]
then $g(t)$ is a continuous function that satisfies
\[
\lim_{t\to \infty}{g(t)} = 1 \h  \lim_{t\to 0}{g(t)} =\infty.
\]

We deduce that for a range of $t$,
\[
f_1(t) - f_2(k(t)) \sim   [\log g(t), 0, 0, 0 ]
\]
which admits every patterns of the form $[x, 0 , 0 , 0]$ with $x>0$.
Similarly
\[
f_2(k(t)) - f_1(t) \sim   [-\log g(t), 0, 0, 0 ]
\]
admits every patterns of the form $[x, 0 , 0 , 0]$ with $x<0$.
This completes the proof.

\end{proof}

\begin{proof}[Proof of Theorem~$\ref{density}$]

From Theorem~$\ref{group}$, there exists a rooted tree $\tau$, a labeling $\psi$ and a vector of real constants $c = (c_1, \ldots, c_S)$ such that
\[
c_s + \log  \sum_{a \in \mathcal{Z}_{i,s} }{ \pi(i) \prod_{(u,v) \in E(\tau)}{P^{uv}_{a_ua_v}(t_{uv})}}= \log (\eta^s_i).
\]
(Recall that $\mathcal{Z}_{i,s}$ is the set of all extensions $a$ of $\psi_s$ to all the nodes of $\tau$ such that $a(\rho)=i$.)
For any $t>0$, we create a new tree $\tau(t)$ by adding an edge $e$ of length $t$ to the root $\rho$ and labeling the additional taxon by the constant vector $(A, A, \ldots, A)$. The log-likelihood function on $e$ of $\tau(t)$ given this taxon labeling is
\begin{eqnarray*}
\ell(t)&=&\sum_{s=1}^S{\log{ \left(\sum_i \sum_{a \in  \mathcal{Z}_{i,s}}{ \pi(i) \prod_{(u,v) \in E(\tau)}{P^{uv}_{a_ua_v}(t_{uv})}} P_{iA}(t)\right)}} \\
&=& - \sum_{s=1}^S{c_s} + \sum_{s=1}^S{\log{\left( \sum_{i}{\eta_i^s P_{i}(t)}\right)}}.
\end{eqnarray*}
Theorem~$\ref{main2}$ implies that the tree $\tau$ can be constructed with at most $64S$ edges. Hence, $\tau(t)$ has at most $(64S +1)$ edges.
\end{proof}

\begin{proof}[Proof of Lemma~\ref{lem1}]
We first consider the case of linear functions.
Assume that $f(x)= a x + b$ such that $f$ is positive in $[0,1]$. We deduce that $b+a =f(1)>0$ .
Hence $f$ can be written as
\begin{eqnarray*}
f(x) &=& ax +b \\
&=& 2(b-a) \left(\frac{1}{4}-\frac{1}{4} x \right)  +(b+a)\left(\frac{1}{4}+\frac{1}{4}x -  \frac{1}{2}x^2 \right) + (b+a)\left(\frac{1}{4}+\frac{1}{4}x +  \frac{1}{2}x^2 \right)\\
&=& 2(b-a) P_3(x) + (b+a) P_2(x) + (b+a) P_1(x)
\end{eqnarray*}
using the transition polynomials $P_i(x)$ from equation $\eqref{equx}$.

Since \{$P_1, P_2, P_3$\} are linearly independent, we deduce that $f$ can be expressed as positive linear combination of $P_1, P_2, P_3$  if and only if $f(-1) =b-a > 0$.

If $f(x)$ is a monic polynomial of degree 2 with no real roots, then $f$ can be written as
\begin{eqnarray*}
f(x)&=& x^2 - 2a x + a^2+b^2\\
&=& [(a-1)^2 + b^2 -1] P_1(x) \\
&& +  [(a-1)^2 + b^2 -2] P_2(x)\\
&&  +2 [(a+1)^2 + b^2 -1] P_3(x).
\end{eqnarray*}
The coefficients are positive if and only if $a\pm bi$ do not belong to $B$.

\end{proof}

\end{document}